\documentclass{article}
\usepackage[utf8]{inputenc}
\usepackage{fullpage}
\usepackage{amsmath, amsthm, amssymb}
\usepackage[table]{xcolor}
\usepackage{bm}
\usepackage{caption, subcaption}
\usepackage[sort&compress, square, numbers]{natbib}
\usepackage{multirow, multicol, booktabs}
\usepackage{soul}
\PassOptionsToPackage{hyphens}{url}
\usepackage{hyperref}
\usepackage{graphicx}
\newtheorem{theorem}{Theorem}[section]
\newtheorem{lemma}[theorem]{Lemma}
\newtheorem{claim}[theorem]{Claim}

\newtheorem{definition}[theorem]{Definition}
\newtheorem{assumption}[theorem]{Assumption}
\newtheorem{remark}[theorem]{Remark}

\newcommand{\AppendixName}[1]{\label{appendix:#1}}
\newcommand{\Appendix}[1]{Appendix~\ref{appendix:#1}}
\newcommand{\AssumptionName}[1]{\label{assumption:#1}}
\newcommand{\Assumption}[1]{Assumption~\ref{assumption:#1}}
\newcommand{\ClaimName}[1]{\label{claim:#1}}
\newcommand{\Claim}[1]{Claim~\ref{claim:#1}}

\newcommand{\DefinitionName}[1]{\label{defn:#1}}
\newcommand{\Definition}[1]{Definition~\ref{defn:#1}}
\newcommand{\EquationName}[1]{\label{eq:#1}}
\newcommand{\Equation}[1]{Eq.~\eqref{eq:#1}}
\newcommand{\FigureName}[1]{\label{fig:#1}}
\newcommand{\Figure}[1]{Figure~\ref{fig:#1}}
\newcommand{\LemmaName}[1]{\label{lemma:#1}}
\newcommand{\Lemma}[1]{Lemma~\ref{lemma:#1}}
\newcommand{\SectionName}[1]{\label{sec:#1}}
\newcommand{\Section}[1]{Section~\ref{sec:#1}}
\newcommand{\TableName}[1]{\label{table:#1}}
\newcommand{\Table}[1]{Table~\ref{table:#1}}
\newcommand{\TheoremName}[1]{\label{thm:#1}}
\newcommand{\Theorem}[1]{Theorem~\ref{thm:#1}}

\newcommand{\bs}{b^*}
\newcommand{\bo}{b^o}
\newcommand{\vs}{v^*}

\newcommand{\rFPA}{\mathbf{rFPA}}

\newcommand{\randauction}{\mathbf{rTruth}}
\newcommand{\eps}{\varepsilon}

\newcommand{\dd}{\mathrm{d}}
\newcommand{\Val}{\textsc{Val}}
\newcommand{\Cost}{\textsc{Cost}}
\newcommand{\Spend}{\textsc{Spend}}

\newcommand{\bme}{\bm{e}}

\newcommand{\alloc}{\pi}
\newcommand{\price}{\mathrm{price}}

\newcommand{\setst}[2]{\left\{ {#1} \,:\, {#2} \right\}}

\newcommand{\LW}{\textsc{LW}}
\newcommand{\OPT}{\textsc{Opt}}
\newcommand{\EQ}{\textsc{Eq}}

\newcommand{\cA}{\mathcal{M}}
\newcommand{\cI}{\mathcal{I}}
\newcommand{\cM}{\mathcal{M}}

\newcommand{\bR}{\mathbb{R}}

\newcommand{\troas}{\text{TROS}}

\title{Efficiency of non-truthful auctions under auto-bidding}
\author{
    Christopher Liaw\thanks{Google, \texttt{\{cvliaw,aranyak,perlroth\}@google.com}},
    Aranyak Mehta\footnotemark[1],
    Andres Perlroth\footnotemark[1]
}

\numberwithin{equation}{section}

\begin{document}

\maketitle
\begin{abstract}
    Auto-bidding is now widely adopted as an interface between advertisers and internet advertising as it allows advertisers to specify high-level goals, such as maximizing value subject to a value-per-spend constraint. Prior research has mostly focused on auctions which are \emph{truthful} (such as SPA) since \emph{uniform} bidding is optimal in such auctions, which makes it manageable to reason about equilibria. A tantalizing question is whether one can obtain more efficient outcomes by leaving the realm of truthful auctions.

This is the first paper to study non-truthful auctions in the prior-free auto-bidding setting. Our first result is that non-truthfulness provides no benefit when one considers deterministic auctions. Any deterministic mechanism has a price of anarchy (PoA) of at least $2$, even for $2$ bidders; this matches what can be achieved by deterministic truthful mechanisms. In particular, we prove that the first price auction has PoA of exactly $2$. For our second result, we construct a randomized non-truthful auction that achieves a PoA of $1.8$ for $2$ bidders. This is the best-known PoA for this problem. The previously best-known PoA for this problem was $1.9$ and was achieved with a truthful mechanism. Moreover, we demonstrate the benefit of non-truthfulness in this setting by showing that the truthful version of this randomized auction also has a PoA of $1.9$. Finally, we show that no auction (even randomized, non-truthful) can improve upon a PoA bound of $2$ as the number of advertisers grow to infinity. 

\end{abstract}

\section{Introduction}
Recent years have seen significant adoption of auto-bidding as an interface for advertisers to bid into internet advertising auctions, including in sponsored search (see, e.g., \cite{goog,fb}). Auto-bidding allows advertisers to express higher-level goals and constraints, as opposed to bids for individual keywords. An auto-bidding agent converts these goals and constraints to per-auction bids in real time. A popular and prototypical auto-bidding setting is that of target-return-on-spend (TROS) where the advertiser expresses the goal of maximizing value with a constraint on the average return on spend (value-per-dollar).
In this setting, we have a set of queries (or ad opportunities), with an auction run separately for each query. The goal of an auto-bidding agent is to place a bid in each auction so as to optimize for the overall \troas~objective for the advertiser.

In this context of automated bidding, current research has started to understand how the design of the underlying auction influences bidding behavior and the efficiency outcomes in system equilibrium when every advertiser uses an auto-bidding agent. In the basic prior free auction model, past work has focused on truthful auctions, that is, auctions where it is optimal for a quasilinear bidder to bid their value (see Sec.~\ref{sec:prel} for details). Even though \troas-bidders are not quasilinear, it turns out that truthfulness is a useful property.
In truthful auctions, a simple \emph{uniform} bidding strategy across queries is optimal \cite{AggarwalBM19}. This is one in which the bid in each auction is proportional to the value. This tractable bidding strategy allows to reason about system equilibria and to bound the loss of welfare. \citet{AggarwalBM19} are the first to mathematically formulate the auto-bidding setting and they show that for the second price auction (SPA) the price-of-anarchy (PoA) is 2, i.e.,
the total value generated in any equilibrium is at least half of the optimal value via some centralized (non auction-based) allocation. Subsequently, \citet{Mehta22} presents a randomized truthful auction, with an improved PoA of approx. 1.89, when there are only two bidders per auction. 

However, contrary to truthful auctions, the PoA for non-truthful auctions is not well-understood.
Understanding the equilibria of non-truthful auctions presents new technical challenges.
In particular, uniform bidding is no longer an optimal bidding strategy. For example, in the first price auction (FPA), a bidder may be able to underbid on some queries with less competition to subsidize more expensive queries where they can overbid. Interestingly, if the advertisers are constrained to bid uniformly in FPA then one can show that the (unique) equilibria has an optimal allocation \cite{DengMMZ21}.
Furthermore, this question of studying the equilibria of non-truthful auctions is of important practical relevance as these auctions are widely used, for example FPA \cite{PaesLemeST20} and generalized second price (GSP) \cite{Varian07}.
Thus, a tantalizing question is whether one can obtain more efficient outcomes by leaving the realm of truthful auctions.

\subsection{Our results}
In this paper, we study the welfare properties of non-truthful auctions in the prior-free auto-bidding setting, thus completing the picture initialized in previous works.

We first show in \Theorem{poa_det_truthful_lb} that for any (non-truthful) deterministic auction, even with 2 bidders, the PoA is no better than 2.
This immediately implies that non-truthfulness provides no additional power over SPA, among deterministic auctions. We complement this by showing, in \Theorem{poa_fpa}, that for FPA, the PoA is precisely 2. Note that this already resolves the motivating practical question of whether the supposedly optimal performance (PoA of 1) of FPA when restricted to uniform bidding persists when we allow the bidders optimal (non-uniform) strategies. 

We next consider if non-truthfulness provides any welfare advantages when coupled with randomization, noting the result from~\citet{Mehta22} that randomization helps improve over SPA in a two bidder setting for truthful auctions.
In contrast to \Theorem{poa_det_truthful_lb}, we show in \Theorem{rFPA_PoA_Eq} that there exists a non-truthful, randomized auction that obtains a PoA of at most $1.8$.
This is currently the best known PoA in this setting.
The auction, which we call the randomized first price auction (rFPA), is relatively simple.
If the higher bid is greater than the lower bid by a prespecified threshold $\alpha > 1$ then the higher bid wins the query outright.
Otherwise, the allocation is randomized depending on how much greater the higher bid is relative to the lower bid.
In either case, the winner pays their bid.
Moreover, the truthful auction obtained by using the same randomized allocation function as rFPA (with the same $\alpha$) but combined with a truthful pricing function, yields a worse PoA of 1.98.\footnote{
    For completeness, we show in \Appendix{randauction_poa} that by choosing the optimal $\alpha$ can improve the PoA of this randomized truthful auction to approximately $1.9$.
}
This shows the power of combining both randomization and non-truthfulness in the design of auctions for auto-bidders.

Our final result is that as the number of bidders $n$ increases to infinity, neither randomization nor non-truthfulness provide any benefit.
In \Theorem{RandPoA}, we prove a lower bound of $2$ for essentially any auction (be it randomized or non-truthful). 
Finally, in \Section{experiments}, we provide experimental results contrasting between SPA, rFPA, and the truthful version of rFPA for synthetic data.

A summary of previous and new results is provided in \Table{results}. We emphasize that all our results for non-truthful auctions are under the much broader set of non-uniform bidding strategies, which contributes in a significant manner and completes the landscape initiated by the previous work.

\begin{table}[t]
\caption{Summary of previous and new results.}
\centering
\begin{tabular}{lllll}
\toprule
  & \multicolumn{2}{c}{\bf Truthful} & \multicolumn{2}{c}{\bf Non-truthful} \\ [0.5ex] 
  \cmidrule(lr){2-3}\cmidrule(lr){4-5}
  & PoA & \multicolumn{1}{c}{Reference} & PoA &\multicolumn{1}{c}{Reference} \\
\midrule
\multirow{2}{*}[-2pt]{Deterministic} & $\leq 2$ & SPA~\cite{AggarwalBM19} & $\leq 2$ & FPA (\Theorem{poa_fpa}) \\
  & $\geq 2$ & \cite[\S 5.3]{AggarwalBM19} &  $\geq 2$ & \Theorem{poa_det_truthful_lb}\\[1ex] 
 Randomized (2 bidders) & $\leq 1.89$ & rTruth \cite{Mehta22} & $\leq 1.8$ & rFPA (\Theorem{rFPA_PoA_Eq}) \\ [1ex] 
 Randomized ($n\rightarrow \infty$ bidders) & $\geq 2$ & \cite[Theorem 5]{Mehta22} & $\geq 2$ & \Theorem{RandPoA} \\ [1ex] 
\bottomrule
\end{tabular}
\TableName{results}
\end{table}

\subsection{Related work}
Prior work on auction design auto-bidding has followed two main threads.
The first thread, which is closest to our work,
assumes no prior information on the advertisers and focuses on auctions that are truthful.
\citet{AggarwalBM19} were the first to study this setting and presented a mathematical formulation for a general auto-bidding product.
As discussed above, they show that uniform bidding is an optimal bidding strategy and, with this, showed that SPA has a PoA of $2$.
Subsequently, \citet{Mehta22} showed that there is a randomized, truthful auction that has a PoA of approximately $1.89$ when there are two bidders and complemented this result by showing that any truthful auction (randomized or otherwise) has a PoA of $2$ as the number of bidders tend to infinity.

\citet{DengMMZ21} initiated a second thread of work in this setting by introducing a model with extra information in order to beat the PoA bound of $2$.
In their model, the auction has access to (a possibly noisy signal of) the individual values of each bidder for each query. They showed that adding a boost to the bid equal to $c$ times the value and then running a SPA leads to a PoA of $\tfrac{c+2}{c+1}$. In a subsequent work, \citet{BalseiroDMMZ21} showed that adding a reserve threshold of $\gamma$ times their values leads to a PoA of $\frac{1}{2-\gamma}$.
We also mention that \citet{BalseiroDMMZ21} look at non-uniform bidding in non-truthful auctions in this model.
Of particular note, they show that setting a reserve of $\gamma \leq 1$ times the advertiser's value yields a PoA of $1/\gamma$ for FPA and GSP when the auto-bidding agents are not restricted to uniform bidding.
In comparison, for FPA without reserves (corresponding to $\gamma = 0$) we show that the PoA for FPA is $2$.
While this model is interesting in its own right, we focus our attention to the pure auction model (as in ~\cite{AggarwalBM19,Mehta22}) in which 
the auction can only work with the bids from the auto-bidders and does not have access to any other information such as the actual values. The pure auction model has practical motivation since the auto-bidder (which converts the values into auction bids) is often a separated system or even a completely different third-party or parties and hence value information may not be available. 

Besides these two main threads, there has also been recent work on understanding the optimal mechanism design for auto-bidders in Bayesian settings~\cite{GolrezaiLP21,BalseiroDMMZ21,BalseiroDMMZ22}; in comparison our work only considers the system response for a given target (as in~\cite{AggarwalBM19,DengMMZ21,Mehta22}).
In addition, there are several lines of work for bidding in the budgeted model, including in non-truthful auctions e.g.,~\cite{BorgsCIJEM07,FeldmanMPS07,BalseiroG19,ConitzerKSM18,ConitzerKPSSMW19}.
\section{Preliminaries}\label{sec:prel}
{\bf Auction basics.}
We let $A$ denote the set of advertisers and $Q$ denote the set of single-slot queries.
Each query $j \in Q$ is sold through an auction mechanism $\cM$, which is common across all queries.
The auction rules are defined by an allocation function $\pi \colon \bR^A_+ \to [0, 1]^A$ where $\sum_{i \in A} \pi_i(b) \leq 1$ for all $b \in \bR^A_+$ and cost function $c \colon \bR^A_+ \to \bR^A_{+}$.
Given bids $b \in \bR^A_+$ where $b_i$ is the bid of advertiser $i$, advertiser $i$ wins the slot (or query) with probability $\pi_i(b)$ for an expected cost $c_i(b)$ (i.e.~advertiser $i$ pays $c_i(b) / \pi_i(b)$ when it wins the query).

An auction is {\em deterministic} if $\pi(b) \in \{0,1\}^A$ whenever there are no ``ties'', i.e.~if $\pi_i(b)>0$ and $\pi_j(b) > 0$ then $b_i = b_j$.
When we allow $\pi(b) \in [0, 1]^A$, we refer to such auctions as {\em randomized} auctions.
An auction is {\em truthful} if and only if $\pi_i(b)$ is monotone in $b_i$ and the cost is given by the formula $c_i(b) = b_i \pi_i(b) - \int_0^{b_i} \pi_i(z, b_{-i}) \, \dd z$ \citep{Myerson81}.\footnote{
    In the quasilinear utility setting, an equivalent definition for an auction to be truthful is that it is optimal for each advertiser to bid their (private) value.
}

\vspace{1em}\noindent\textbf{Auto-bidding basics.}
We consider the setting where an auto-bidding agent submits bids on behalf of advertiser $i \in A$ to maximize the advertiser's value subject to a prespecified target return-on-spend (ROS), $T_i$.
Mathematically, the auto-bidding agent for advertiser $i$ aims to solve the following optimization problem:
\begin{equation}
\EquationName{tros}
\begin{aligned}
    \text{maximize: \quad}   & \sum_{j \in Q} \pi_{i, j} v_{i, j} \\ 
    \text{subject to: \quad} & \sum_{j \in Q} \pi_{i, j} c_{i, j}(\pi_{i, j}) \leq T_i \sum_{j \in Q} \pi_{i, j} v_{i, j} \\
                            & 0 \leq \pi_{i, j} \leq 1 \quad \forall j \in Q,
\end{aligned}
\end{equation}
where $\pi_{i, j}$ is the decision variable that gives the probability that advertiser $i$ wins query $j$,
$c_{i,j}(\pi_{i,j})$ is the cost to advertiser $i$ if the advertiser wins query $j$ with probability $\pi_{i, j}$ (which is determined by the auction rules),
$T_i$ is the advertiser's target ROS,
and $v_{i, j}$ is the value that advertiser $i$ has for query $j$.
Note that the bid that advertiser $i$ places on query $j$ is implicitly determined by $\pi_{i, j}$ and together determine $c_{i, j}$.

In our auto-bidding setting, the goal is maximize a notion of welfare known as the \emph{liquid welfare} which captures an advertiser's willingness to pay for a certain allocation.
This notion was introduced by \citet{DobzinskiL14} and has been widely adopted in the auto-bidding literature \cite{Mehta22,AggarwalBM19,DengMMZ21, BalseiroDMMZ21}.
\begin{definition}[Liquid welfare]
    For an allocation $\{\pi_{i, j}\}_{i \in A, j \in Q}$, the liquid welfare is defined to be the quantity $\LW(\{\pi_{i,j}\}) = \sum_{i \in A} T_i \sum_{j \in Q} \pi_{i, j} v_{i,j}$.
\end{definition}
Notice that for TROS bidders, the Liquid Welfare corresponds to the total target-weighted value obtained in a given allocation.

Following \citet{Mehta22}, we consider the following solution concept for equilibrium.
\begin{definition}[$\gamma$-equilibrium]
    Let $\gamma \geq 0$ and fix an auction $\cA$ with advertisers $A$ and queries $Q$.
    Let $\{b_{i,j}\}$ be the set of bids and let $\{\pi_{i, j}\}, \{c_{i,j}\}_{i, j}$ be the resulting allocations and costs.
    We say that the bids $\{b_{i, j}\}$ are in $\gamma$-equilibrium if each bidder satisfies its ROS constraint, i.e.~$\sum_{j \in Q} \pi_{i, j} c_{i, j} \leq T_i \sum_{j \in Q} \pi_{i, j} v_{i, j}$
    and the following statement holds.
    Suppose that advertiser $i$ deviates and changes its bids from $(b_{i,j})_{j \in Q}$ to $(b_{i,j}')_{j \in Q}$
    while all other advertisers retain their bids.
    Let $\{\pi_{i, j}'\}$ and $\{c_{i, j}'\}$ be the new allocations and costs.
    Then either
    \begin{enumerate}
        \item $\sum_{j \in Q} \pi_{i, j}' v_{i, j} \leq (1+\gamma)\sum_{j \in Q} \pi_{i, j} v_{i, j}$ (advertiser $i$ gains no more than $(1+\gamma)$ multiplicative value); or
        \item $\sum_{j \in Q} \pi_{i, j}' c_{i, j}' > T_i \sum_{j \in Q} \pi_{i, j}' v_{i, j}$ (advertiser $i$ exceeds its ROS constraint).
    \end{enumerate}
    We refer to $0$-equilibrium as simply equilibrium.
\end{definition}

\begin{definition}[(Liquid) Price of Anarchy]
    Fix an auction $\cA$.
    For an instance $\cI$, which denotes a set of advertiser, queries, and their values, let $\pi^{\OPT}$ denote an allocation that maximizes the liquid welfare and let $\Pi^{\EQ}$ denote the set of allocations that form an equilibrium.
    Then the PoA of the auction $\cA$ is defined as $\sup_{\cI} \sup_{\pi^{\EQ} \in \Pi^{\EQ}} \frac{\LW(\pi^{\OPT})}{\LW(\pi^{\EQ})}$.
\end{definition}
The following observation allows us to reduce some notation in this paper.
\begin{remark}
    Without loss of generality, we assume $T_i = 1$ for all advertisers $i$ in the remainder of the proof.
    This is valid since one can replace $v_{i, j}$ with $T_i v_{i, j}$ for all advertisers $i$ and queries $j$.
\end{remark}
\section{Lower bound on PoA for deterministic auctions}
\SectionName{lb_deterministic}
In this section, we prove that the worst-case PoA for any deterministic auction is at least $2$,
even for two bidders, provided that the auction satisfies a couple of natural assumptions.\footnote{
This clearly implies the same PoA for $n$ advertisers either by duplicating the advertisers or introducing ``dummy'' advertisers with zero values.}
\begin{assumption}
\AssumptionName{det}
We assume that the auction $\cA$ satisfies the following assumptions.
\begin{enumerate}
\item \textbf{Allocation.} 
(i) The allocation function $\pi_1(b_1, b_2)$ (resp.~$\pi_2(b_1, b_2)$) is non-decreasing in $b_1$ (resp.~$b_2$), and (ii) the winner is always the highest bidder, i.e.~if $\pi_1(b_1, b_2) = 1$ then $b_1 > b_2$ (and analogously when $\pi_2(b_1, b_2) = 1$).
\item \textbf{Pricing.} 
Let $\price(b_1,b_2)$ be the price that the winner pays. We assume that (i) $\price(b_1, b_2)$ is non-decreasing in $b_1, b_2$, and (ii) satisfies that $\lim_{b \to \infty} \price(b, b) = \infty$.
\end{enumerate}
\end{assumption}

\begin{remark}
    The assumption $\lim_{b \to \infty} \price(b, b) = \infty$ ensures that the prices are not capped.
    Such an assumption is \emph{necessary} to guarantee that an
    equilibrium exists in the auction.\footnote{Consider an auction such that $\lim_{b \to \infty} \price(b, b) \leq M$ for some $M \geq 0$ (e.g., a third-price auction with two bidders). Observe that
    for all $b_1 \geq b_2$, $\price(b_1, b_2) \leq \price(b_1, b_1) \leq M$ (and similarly for $\price(b_2, b_1)$). 
    Thus, if there are two bidders with value strictly more than $M$ then no
    equilibrium exists since the two bidders can always outbid each other and pay at most $M$.}
\end{remark}

\begin{theorem}\label{th:1}
\TheoremName{poa_fpa_lb}
\TheoremName{poa_det_truthful_lb}
No auction that satisfies \Assumption{det} can have a PoA better than $2$ even for two bidders.
\end{theorem}
To complement \Theorem{poa_fpa_lb}, we show that a first price auction has a price of anarchy of at most $2$ under non-uniform bidding but we relegate the proof to \Appendix{poa_fpa}.
\begin{theorem}
    \TheoremName{poa_fpa}
    In a first price auction, the liquid welfare of any equilibrium is at least half of the optimal welfare.
\end{theorem}
\begin{proof}[Proof of \Theorem{poa_det_truthful_lb}]
Observe that we can focus on auctions where there is a bid $B_1$ such that if one of the bidders bids $B_1$ (say bidder $1$) and the other bidder bids $0$ then bidder $1$ wins the item (and pays $\price(B_1, 0)$). If such $B_1$ does not exist (i.e.~$\pi_1(b, 0) = 0$ for all $b$), clearly the PoA is infinity (consider the instance with a single query, bidder $1$ has value $1$ and bidder $2$ has value $0$).
    
We now construct an instance as follows. Fix $\gamma \geq 0$ and $\eps < \gamma$.
Let $B_1$ be as above and let $B_2$ be arbitrary which we will take to infinity later.
Consider two queries:
\begin{enumerate}
    \item In query $1$, advertiser $1$ has value $\price(B_2, B_2) + \price(B_1, 0)$; advertiser $2$ has value $0$.
    \item In query $2$, advertiser $1$ has value $\eps \cdot \price(B_1, 0)$; advertiser $2$ has value $\price(B_2, B_2) - \eps$.
\end{enumerate}
We claim that the following set of bids is a $\gamma$-equilibrium: advertiser $1$ bids $B_1$ on query $1$ and $B_2$ on query $2$, while advertiser $2$ bids $0$ on both queries.
    
First, note that advertiser $1$ meets its ROS constraint and it receives value at least $\price(B_2, B_2) + \price(B_1, 0)$
(from the first item alone) and pays at most $\price(B_2, 0) + \price(B_1, 0) \leq \price(B_2, B_2) + \price(B_1, 0)$
(the upper bound follows from the assumption that $\price(b_1, b_2)$ is non-decreasing in $b_1, b_2$.
In addition, bidder $1$ cannot increase its value by more than $\gamma$ since it has value $\eps \cdot \price(B_1, 0) < \gamma (\price(B_2, B_2) + \price(B_1, 0))$ for query $2$.

Next, advertiser $2$ can only win query $2$ if it bids at least $B_2$ but then the price would be strictly more than its value.
So advertiser $2$ also has no incentive to deviate.

The liquid welfare of this equilibrium is at most $\price(B_2, B_2) + \price(B_1, 0) + \eps \cdot \price(B_1, 0)$
whereas the optimal allocation has liquid welfare equal to $2\cdot\price(B_2, B_2) + \price(B_1, 0) - \eps$.
Since $B_2, \eps$ is arbitrary, we can take $B_2 \to \infty$ (in which case, $\price(B_2, B_2) \to \infty$) and $\eps \to 0$ so that the ratio approaches $2$.
\end{proof}

\section{Lower bound on PoA for randomized auctions}
\SectionName{lb_randomized}
In this section, we show that the PoA of any randomized auction is at least $2$ provided that the auction satisfies some natural assumptions.

\begin{assumption}
\label{assumption}
Suppose there are $n$ bidders.
We assume an auction $\cA$ satisfying the following properties.
\begin{enumerate}
    \item \emph{\textbf{Allocation.}}
    For $i \in [n]$, let $\alloc_i(b_1, \ldots, b_n) \in [0,1]$ be the probability that bidder $i$ wins the item.
    We assume that $\alloc_i(b_1, \ldots, b_n)$ is non-decreasing in $b_i$ and that $\sum_{i \in [n]} \alloc_i(b_1, \ldots, b_n) \leq 1$.
    \item \emph{\textbf{Pricing.}}
    For $i \in [n]$, let $\price_i(b_1, \ldots, b_n) \in \bR_{\geq 0}$ be the price that bidder $i$ pays if bidder $i$ wins the item (and otherwise, pays $0$).
    We assume that $\price_i(b_1, \ldots, b_n)$ is non-decreasing in $b_j$ for all $j \in [n]$.
    Further, we assume that $\lim_{b \to \infty} \price_i(b\cdot \bme_S) = +\infty$ whenever $i \in S$ and $|S| \geq 2$.
    \item \emph{\textbf{Anonymity.}}
    Let $b \in \bR^{n}_{\geq 0}$.
    For $i, j \in [n]$, let $b^{i \leftrightarrow j}$ be the vector such that $b^{i \leftrightarrow j}_i = b_j$ and $b^{i \leftrightarrow j}_j = b_i$.
    Then we assume that $\alloc_i(b) = \alloc_j(b^{i \leftrightarrow j})$ and $\price_i(b) = \price_j(b^{i \leftrightarrow j})$.
\end{enumerate}
\end{assumption}
The first assumption means that an advertiser will have a higher chance of winning as their bid increases.
The second assumption means that as the bids increase, so do the prices and, if multiple bids become very large then so do the prices.
The third assumption means that the auction should be anonymous: if we swap the bids of two advertisers then the outcome should also be swapped.

The main result in this section is the following theorem which states that it is not possible to obtain a PoA better than $2$ for any auction.
The proof can be found in \Appendix{RandPoA} and is an extension of the techniques from \Section{lb_deterministic}.
\begin{theorem}
    \TheoremName{RandPoA}
    For every $\gamma, \delta > 0$, $k \geq 1$, and auction $\cA$, there exists an instance with $2k$ bidders
    in which a $\gamma$-equilibrium exists whose total value is less than $\frac{k+2}{2(k+1)} \cdot (1+\delta)$ fraction of the optimal value.
\end{theorem}

\section{Randomized FPA}
\SectionName{rfpa}

In this section, we show that a randomized version of the first price auction can be used to obtain better equilibrium outcomes in auto-bidding for two bidders.
Note that, in the following definition, $\alpha = 1$ reduces to the usual first price auction.
\begin{definition}[$\rFPA(\alpha)$]
\DefinitionName{rfpa}
We define the allocation of $\rFPA(\alpha)$ as follows.
\begin{itemize}
    \item If $b_1 = \beta b_2$ for $\beta \in [\frac{1}{\alpha}, \alpha]$ then allocate to advertiser $1$ with probability $\frac{1}{2} \left( 1 + \ln_\alpha \beta \right)$ and advertiser $2$ with probability $\frac{1}{2}\left( 1 - \ln_\alpha \beta \right)$.
    \item Otherwise, allocate to the higher bidder with probability $1$.
\end{itemize}
In both cases, the winner is charged their bid.
\end{definition}

\begin{definition}[Undominated bids]
Let $\mathbf{b}_i = (b_{i,1}, \ldots, b_{i,m})$ be the bids of advertiser $i$ for queries $1, \ldots, m$.
We say that $\mathbf{b}_i$ is \emph{undominated} if advertiser $i$'s ROS constraint are satisfied and
if for all $j$, changing the bid for $j$ from $b_{i, j}$ to $b_{i, j}'$
then either
(i) advertiser $i$'s value from query $j$ does not increase or
(ii) advertiser $i$'s ROS constraint is violated.
For a collection of bids $\mathbf{b} = (\mathbf{b}_1, \ldots, \mathbf{b}_n)$,
we say that $\mathbf{b}$ is undominated if $\mathbf{b}_i$ is undominated for all $i \in [n]$.
\end{definition}
\noindent The set of undominated bids contains the set of equilibrium bids as a subset.
However, the set of undominated bids may be strictly larger since it only guarantees that the advertiser will not deviate on a \emph{single} query.
For example, an advertiser may decrease their bid on one query while increasing their bid on another query to maintain their ROS constraint
and increase their total value.

Our main result in this section is the following theorem.
\begin{theorem}
    \TheoremName{rFPA_PoA}
    For any set of undominated bids for two bidders, $\rFPA(\alpha=1.4)$ obtains at least $1/1.8$-fraction of the optimal welfare.
\end{theorem}
The following theorem is a trivial corollary of \Theorem{rFPA_PoA}.
\begin{theorem}
    \TheoremName{rFPA_PoA_Eq}
    For two bidders, the PoA of $\rFPA(\alpha=1.4)$ is at most $1.8$.
\end{theorem}

\noindent {\bf High-level proof overview.}
The goal in the proof is to show that, even on queries where there is a misallocation, we can still use these queries to prove a lower bound on the liquid welfare.
Indeed, the main starting observation is that, because of the ROS constraint, the \emph{spend} of the two bidders is a lower bound on the liquid welfare.
Next, we consider two main cases where there is misallocation.

In one extreme case, if rFPA allocates to the wrong bidder with probability $1$ (complete misallocation) then we show, in \Lemma{wp1_win_lb}, then the winner's bid (and thus, their spend) is very high.
This allows us to still extract some contribution to the liquid welfare.
The reason that the winning bid is very high is because if this was not so, the ``right'' bidder would deviate by bidding their own value.

In the other case, when the query is shared between the two bidders, we show that there is a balance between how likely rFPA allocates to the ``wrong'' bidder and how much the bidders are spending in this query.
In fact, the more that rFPA misallocates, the greater is the spend (\Lemma{defection}).

We note that somewhat similar techniques are used in the analysis of the randomized truthful auction studied by \citet{Mehta22}.
However, since their auction is truthful, it can be seen that bidders always bid at their value on all queries.
This is not the case in our setting since bidders can potentially underbid; however as noted above, we can bound the degree of underbidding in the critical cases.

We now present two key lemmata whose proofs are deferred to \Appendix{wp1_win_lb_defection}.

\begin{lemma}
    \LemmaName{wp1_win_lb}
    Fix a query $j$ and suppose that advertiser $1$ wins query $j$ with probability $1$.
    If the bids are undominated then advertiser $1$ bids and pays at least $\alpha v_{2, j}$.
    The same statement holds with advertisers $1$ and $2$ swapped.
\end{lemma}
\begin{lemma}
    \LemmaName{defection}
    For any set of undominated bids, if $\pi_{i, j} \in (0, 1)$ then $b_{i, j} \geq \frac{v_{i, j}}{1 + \ln(\alpha) + \ln(\beta_i)}$
    where $\beta_i = b_{i, j} / b_{2-i, j}$ is the ratio between advertiser $i$'s bid and the other advertiser's bid.
\end{lemma}
The following technical lemma proves an upper bound on the PoA for $\rFPA(\alpha)$ for every $\alpha > 1$.
\begin{lemma}
    \LemmaName{rfpa_technical}
    For any set of undominated bids for two bidders, $\rFPA(\alpha)$ obtains at least $f(\alpha)$ fraction of the optimal welfare where
    \begin{equation}
        \EquationName{f_alpha}
        f(\alpha) = \max_{\substack{\gamma, \eta \geq 0 \\ \gamma + \eta = 1}} \min \left\{
            \eta \alpha,
            \gamma,
            \min_{\beta \in [1/\alpha, \alpha]} g(\alpha, \beta, \gamma, \eta)
        \right\} 
    \end{equation}
    and
    \begin{equation}
        \EquationName{g_alpha_beta_gamma_eta}
        g(\alpha, \beta, \gamma, \eta) =
        \frac{\gamma}{2} \left( 1 + \ln(\beta) / \ln(\alpha) \right) +
        \frac{\eta(1 + \ln(\beta) / \ln(\alpha))}{2(1+\ln(\alpha)+\ln(\beta))} +
        \frac{\eta(1 - \ln(\beta) / \ln(\alpha))}{2\beta(1+\ln(\alpha)+\ln(\beta))}.
    \end{equation}
\end{lemma}
\begin{proof}
The structure of the proof is roughly similar to that in \cite{Mehta22}.
Fix a query and let $i^* \in \{1, 2\}$ be the advertiser that wins the query $j$ in the optimal allocation and let $i^o = 2 - i^*$ refer to the other bidder.
Let $\bs_j$, $\bo_j$ be their respective bids.
We consider three classes of sets depending on $\bs_j / \bo_j$.

\noindent {{\bf Case 1:} $\bs_j / \bo_j < 1/\alpha$.}
Let $Q_1 = \setst{j \in Q}{\bs_j / \bo_j \leq 1/\alpha}$.
Define $m_1$ as the probability that the optimal bidder is matched to query $j$ whenever $j \in Q_1$.
Furthermore, let $\Spend_1(j)$ be the spend in this query.
Then $m_1 = 0$ and $\Spend_1(j) = \bo_j \geq \alpha \vs_j$ where the inequality is from \Lemma{wp1_win_lb}.

\noindent {{\bf Case 2:} $1/\alpha \leq \bs_j / \bo_j \leq \alpha$.}
For $\beta \in [1/\alpha, \alpha]$, let $Q_2^\beta = \setst{j \in Q}{\bs_j / \bo_j = 1/\beta}$.
Define $m_2^\beta$ as the probability of matching $i^*$ to query $j$ whenever $j \in Q_2^\beta$ and let $\Spend_2^{\beta}(j)$ be the expected spend on query $j$.
Then $m_2^\beta = \frac{1}{2}\left( 1 + \frac{\ln \beta}{\ln \alpha} \right)$ and
$\Spend_2^\beta(j)
= m_2^\beta \cdot \bs + (1 - m_2^\beta) \cdot \bo
= m_2^\beta \cdot \bs + (1 - m_2^\beta) \cdot \frac{\bs}{\beta}
\geq m_2^\beta \cdot \frac{\vs}{1+\ln(\alpha)+\ln(\beta)} + (1 - m_2^{\beta}) \cdot \frac{\vs}{\beta(1+\ln(\alpha)+\ln(\beta))}
$,
where the last inequality is from \Lemma{defection}.

\noindent {{\bf Case 3:} $\alpha < \bs_j / \bo_j$.}
Finally, let $Q_3 = \setst{j \in Q}{\bs_j / \bo_j \geq \alpha}$.
For any $j \in Q_1$, define $m_3$ as the probability that the optimal bidder is matched in $j$, and $\Spend_3(j)$ be the expected spend in $j$. Then $m_3 = 1$ and $\Spend_3(j) \geq 0$.

Since the number of queries $Q$ is finite, there exists a finite list $\beta_1, \ldots, \beta_k$ such that
$Q = Q_1 \cup Q_3 \cup \left( \bigcup_{\ell=1}^k Q_{2}^{\beta_\ell} \right)$.
Now observe that the ROS constraint implies that
\begin{align*}
    \LW
    & \geq \sum_{j \in Q_1} \Spend_1(j) + \sum_{\ell=1}^{k} \sum_{j \in Q_2^{\beta_\ell}} \Spend_2^{\beta}(j) + \sum_{j \in Q_3} \Spend_3(j) \\
    & \geq \sum_{j \in Q_1} s_1 \vs_j + \sum_{\ell=1}^k \sum_{j \in Q_2^{\beta_{\ell}}} s_2^{\beta_{\ell}} \vs_j
    + \sum_{j \in Q_3} s_3 \vs_j
\end{align*}
where $s_1 = \alpha$ and $s_2^\beta = \frac{1+\frac{\ln \beta}{\ln \alpha}}{2(1+\ln\alpha+\ln\beta)} + \frac{1-\frac{\ln \beta}{\ln \alpha}}{2\beta(1+\ln \alpha+\ln\beta)}$, and $s_3 = 0$.
Indeed, in the first inequality, the RHS is the total spend and the LHS is the total liquid welfare
(recall that if the bids are undominated then the liquid welfare of each advertiser is an upper bound on the spend of each bidder).
Next, observe that we can lower bound the liquid welfare by
\begin{align*}
    \LW
    & \geq \sum_{j\in Q_1} m_1 \vs_j + \sum_{\ell=1}^k \sum_{j \in Q_2^{\beta_\ell}} m_2^{\beta_\ell} \vs_j + \sum_{j \in Q_3} m_3 \vs_j,
\end{align*}
where $m_1, m_2^\beta, m_3$ are as described above.
Finally, let $\OPT = \sum_{j \in Q} \vs_j$ be the value of the optimal allocation.
Let $x_1 = \sum_{j \in Q_1} \vs_j / \OPT$, $x_2^\beta = \sum_{j \in Q_2^{\beta}} \vs_j / \OPT$, and $x_3 = \sum_{j \in Q_3} \vs_j / \OPT$.
Thus,
\[
    \frac{\LW}{\OPT} \geq \max\left\{
        s_1 x_1 + \sum_{\ell=1}^{k} s_2^{\beta_{\ell}} x_{2}^{\beta_{\ell}} + s_3 x_3,
        m_1 x_1 + \sum_{\ell=1}^{k} m_2^{\beta_{\ell}} x_{2}^{\beta_{\ell}} + m_3 x_3
    \right\}.
\]
Note that $x_1 + \sum_{\ell=1}^k x_2^{\beta_\ell} + x_3 = 1$ so to find the optimum ratio of $\LW/\OPT$, it suffices to solve the following optimization problem, where $\mu_2$ is some finite measure on $[1/\alpha, \alpha]$.
$$\inf~ z$$
$$z\geq \int m_2^\beta \, \dd \mu_2(\beta) + x_3 m_3$$
$$z\geq x_1 s_1 + \int s_2^\beta \, \dd \mu_2(\beta)$$
$$x_1 + \mu_2([1/\alpha, \alpha]) + x_3 \geq 1$$
A standard duality argument shows that the above program is bounded below by
\[
\max_{\substack{\gamma, \eta \geq 0 \\ \gamma + \eta = 1}}
\min \left\{ \eta \alpha, \gamma,
\min_{\beta \in [1/\alpha, \alpha]} \gamma m_2^\beta + \eta s_2^\beta
\right\}
\]
which is exactly \Equation{f_alpha} and \Equation{g_alpha_beta_gamma_eta}.
\end{proof}

\begin{proof}[Proof of \Theorem{rFPA_PoA}]
    We take $\gamma = 0.56$ so $\eta = 0.44$. Since the function is complicated, we can plot the three terms in the definition of \Equation{f_alpha} in \Figure{rfpa_approx} below.
    \begin{figure}[ht]
        \centering
        \includegraphics[width=0.7\textwidth]{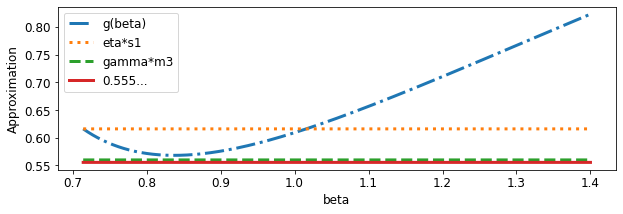}
        \caption{\footnotesize Plot of the three components of $f(\alpha)$ as defined in \Equation{f_alpha} when $\alpha=1.4$.
            The dashed-dotted blue line is $g(\beta)$ with $\alpha, \eta, \gamma$ fixed as defined in \Equation{g_alpha_beta_gamma_eta}.
            The dotted orange line is $\eta s_1$ and the dashed green line is $\gamma m_3$.
            The solid red line at the bottom is $1/1.8 = 0.555\ldots$.
            Thus, $f(\alpha=1.4) \geq 1/1.8$.
        }
        \FigureName{rfpa_approx}
    \end{figure}
\end{proof}

\noindent \textbf{Comparison with truthful pricing.}
Here, we compare the efficiency between $\rFPA(\alpha)$ and its truthful variant which we call $\randauction(\alpha)$.
The latter auction is a continuous version of the one used in \citet{Mehta22}.
For completeness, we prove that it has a PoA of 1.896 in \Appendix{randauction_poa}.
\begin{definition}[$\randauction(\alpha)$]
\DefinitionName{Rand_alpha}
We define the allocation and price function of $\randauction(\alpha)$ (for $\alpha \geq 1$) as follows (see \Appendix{Rand_alpha_calculation} for the calculation of the prices).
\begin{itemize}
\item If $b_1 = \beta b_2$, for $\beta \in [\frac{1}{\alpha}, \alpha]$, then
\begin{itemize}
    \item bidder 1 wins with 
probability $\frac{1}{2}(1 + \ln_\alpha \beta)$, and pays an expected price of $b_2 \cdot \frac{\beta - 1/\alpha}{2 \ln \alpha}$.
    \item bidder 2 wins with probability $\frac{1}{2}(1 - \ln_\alpha \beta)$, and pays an expected price of $b_1 \cdot \frac{1/\beta - 1/\alpha}{2 \ln \alpha}$.
\end{itemize}
\item If $b_1 \geq \alpha b_2$: bidder 1 wins with probability 1 and has a price of $b_2 \cdot \frac{\alpha-1/\alpha}{2 \ln \alpha}$.
\end{itemize}
\end{definition}
\begin{claim}
    \ClaimName{rand_poa_lb}
    For any $\alpha \geq 1$, $\randauction(\alpha)$ has a PoA of at least $1 + \frac{2 \ln(\alpha)}{\alpha - 1/\alpha}$.
\end{claim}
The proof of \Claim{rand_poa_lb} can be found in \Appendix{rand_poa_lb}.
In particular, using the same $\alpha = 1.4$ as in \Theorem{rFPA_PoA} shows that $\randauction(\alpha)$ has a worst-case PoA of at least $1.98$ despite having exactly the same allocation function as $\rFPA(\alpha)$.

\section{Experiments}
\SectionName{experiments}
In this section, we provide some experimental results to compare and contrast between SPA, rFPA (defined in \Definition{rfpa}), and rTruth (defined in \Definition{Rand_alpha}).
The purpose of these experiments is only to visually illustrate the utility of randomness and non-truthfulness in the auto-bidding setting.

\textbf{Simulation setup.}
We generate the synthetic data in four different ways.
\begin{itemize}
    \item[(a)] For the first set of experiments, there are 50 queries and the advertisers' values for each query are drawn uniformly at random from $[0.3,1]$.
    \item[(b)] For the second set of experiments, there are also 50 queries and the advertisers' values are drawn as follows. With probability $1/2$, an advertiser's values are drawn from $[1, 1.2]$ and with the remaining probability, the values are drawn from $[0.3, 0.5]$. Note that the advertisers' values are independent from each other.
    \item[(c)] For the third set of experiments, the advertisers' values are deterministic and there are two queries.
    Advertiser $1$ has values $1$ and $0.01$ for the queries, respectively
    and advertiser $2$ has values $0.01$ and $0.99$ for the queries, respectively.
    \item[(d)] For the fourth set of experiments, for completeness, we contrive a simple example where we expect randomization to be detrimental.
    There is only a single query where advertiser $1$ has value $1$ and advertiser $2$ has value $0.9$.
\end{itemize}
For the first two experiments, we run each experiment $20$ times and report the expected results.

\noindent \textbf{Computing best response and equilibrium.}
To compute an equilibrium for the two advertisers, we iteratively compute a best response for each advertiser until convergence.
Since uniform bidding is the optimal bidding strategy in a truthful auction \cite{AggarwalBM19}, it is straightforward to compute a best response in SPA and rTruth.
To compute a best response for rFPA is slightly more challenging as uniform bidding is not optimal.
Fortunately, computing the best response turns out to be a convex program and so can be solved relatively efficiently.\footnote{It is convex program assuming that the bids are within $\alpha$ factor of each other. This is without loss of generality since an advertiser never has incentive to bid $1/\alpha$ factor less than or $\alpha$ factor more than the other bid.}

\noindent \textbf{Simulation results.}
\Figure{experiment_results} shows the simulation results.
In \Figure{uniform} and \Figure{four_parts} we see that the randomized first price auction with an appropriately tuned $\alpha$ outperforms rTruth and the SPA.
In \Figure{z_graph}, we see an example where randomization is extremely important in order to achieve good welfare.
Finally, in the last example in \Figure{one_high_one_low}, the instance was constructed so that randomization is actually (slightly) detrimental.
Recall that in this instance, there are two relatively high value advertisers and a single query.
The SPA picks the correct advertiser but randomization may sometimes swap if the two values are not too far from each other.
\begin{figure}[hb]
     \centering
     \begin{subfigure}[b]{0.47\textwidth}
         \centering
         \includegraphics[width=\textwidth]{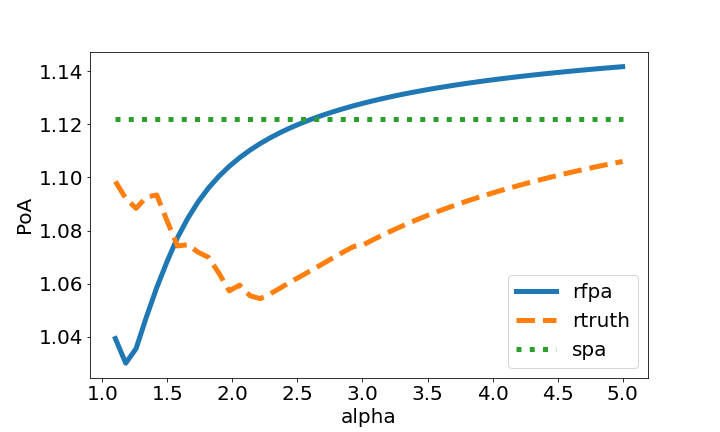}
         \caption{}
         \FigureName{uniform}
     \end{subfigure}
     \hfill
     \begin{subfigure}[b]{0.47\textwidth}
         \centering
         \includegraphics[width=\textwidth]{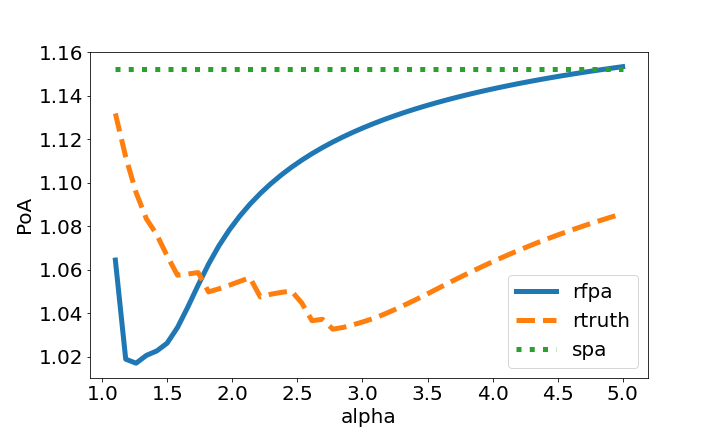}
         \caption{}
         \FigureName{four_parts}
     \end{subfigure}
     \\ \vspace{-1.4em}
     \begin{subfigure}[b]{0.47\textwidth}
         \centering
         \includegraphics[width=\textwidth]{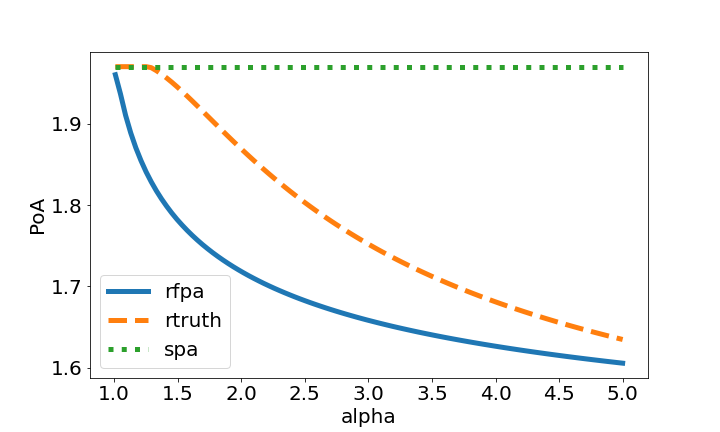}
         \caption{}
         \FigureName{z_graph}
     \end{subfigure}
     \hfill
     \begin{subfigure}[b]{0.47\textwidth}
         \centering
         \includegraphics[width=\textwidth]{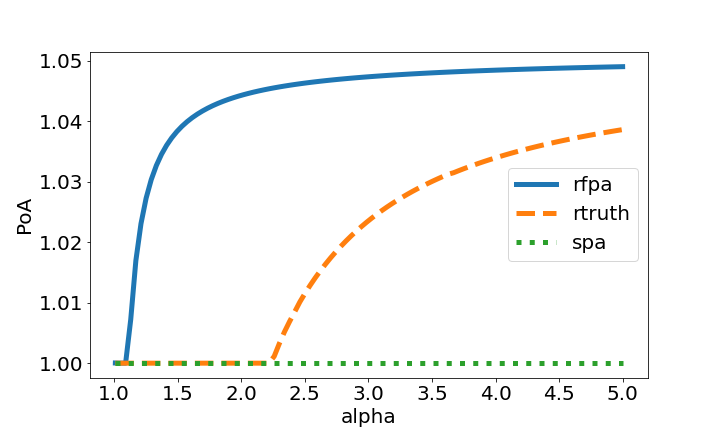}
         \caption{}
         \FigureName{one_high_one_low}
     \end{subfigure}
    \caption{\footnotesize
        Simulation results contrasting the randomized first price, randomized truthful, and SPA in an auto-bidding setting. (Note: the $y$-axis is the PoA so lower is better.)
    }
    \FigureName{experiment_results}
\end{figure}

\bibliographystyle{plainnat}
\bibliography{refs}

\appendix
\section{Proof of \Theorem{poa_fpa}}
\AppendixName{poa_fpa}
\newcommand{\vfp}{v^{\textsc{FP}}}
\begin{proof}[Proof of \Theorem{poa_fpa}]
    Let $Q$ be the set of queries and $A$ be the set of advertisers.
    For $j \in Q$, let $\vfp_j$ be the value that the winner of query $j$ has for query $j$ and let $v^*_j = \max_{i \in A} v_{i, j}$ be the highest value for query $j$.
    Let $S_1 = \setst{j}{\vfp_j = v^*_j}$ and $S_2 = \setst{j}{\vfp_j < v^*_j}$.
    Finally, for $j \in Q$, let $p_j$ be the price that was paid for query $j$.
    
    First, we observe that we have
    \begin{equation}
        \EquationName{fpa_1}
        \sum_{j \in Q} \vfp_j
        ~\geq~\sum_{j \in S_1} \vfp_j
        ~=~\sum_{j \in S_1} v^*_j.
    \end{equation}
    Next, observe that for $j \in S_2$, we must have $p_j \geq v^*_j$.
    Indeed, if $p_j < v^*_j$ then the highest value bidder may bid $v^*_j$ to win query $j$ without violating their ROS constraint.
    Thus, we have
    \begin{equation}
        \EquationName{fpa_2}
        \sum_{j \in Q} \vfp_j
        ~\geq~\sum_{j \in Q} p_j
        ~\geq~\sum_{j \in S_2} p_j
        ~\geq~\sum_{j \in S_2} v^*_j,
    \end{equation}
    where the first inequality is because all bidders must satisfy their ROS constraint.
    Combining \Equation{fpa_1} and \Equation{fpa_2}, we get that
    \[
        2 \sum_{j \in Q} \vfp_j
        ~\geq~\sum_{j \in Q} v^*_j,
    \]
    which proves the theorem.
\end{proof}

\section{Missing proofs from \Section{rfpa}}
\AppendixName{rfpa}
\subsection{Proofs of \Lemma{wp1_win_lb} and \Lemma{defection}}
\AppendixName{wp1_win_lb_defection}
\begin{proof}[Proof of \Lemma{wp1_win_lb}]
    Suppose that advertiser $1$ wins query $j$ with probability $1$ but places a bid $b_{1,j} < \alpha v_{2, j}$.
    Then we claim that the advertisers are not in equilibrium.
    Indeed, advertiser $2$ would then be able to bid $v_{2,j}$ and win with probability $\pi_{2, j} = \frac{1}{2}\left( 1 + \ln_{\alpha} (v_{2,j} / b_{1,j}) \right) > 0$ because $v_{2,j} / b_{1,j} > 1/\alpha$.
    Thus, advertiser $2$ would increase their value by $\pi_{2,j} v_{2,j}$.
    Furthermore, advertiser $2$'s ROS constraint continues to be satisfied because advertiser $2$ pays $\pi_{2,j} v_{2,j}$ for this query.
    Thus, in equilibrium, it must be that $b_{1, j} \geq \alpha v_{2,j}$.
\end{proof}

\begin{proof}[Proof of \Lemma{defection}]
    We prove this for $i = 1$ as the argument is identical for $i = 2$.
    We also drop the subscript $j$ since the query is fixed in the proof and we let $\beta = \beta_1$.
    
    Overall, the proof is straightforward.
    We show that if $b_1 < \frac{v_1}{1 + \ln(\alpha) + \ln(\beta)}$ then advertiser $1$'s bids are not undominated
    and that increasing their bid on the present query strictly increases their value while ensuring that the ROS constraint is satisfied.
    For $b \in [1/\alpha, \alpha] b_2$, let $\pi(b) = \pi_1(b, b_2) = \frac{1}{2} \left( 1 + \frac{ \ln(b/b_2) }{\ln \alpha } \right)$.
    Let $\Val(b) = v_1 \cdot \pi(b)$ and $\Cost(b) = b \cdot \pi(b)$.
    To prove the claim, it suffices to show that if $b_1 < \frac{v_1}{1 + \ln(\alpha) + \ln(\beta)}$ then $\frac{\dd}{\dd b} \Val(b_1) > \frac{\dd}{\dd b} \Cost(b_1)$,
    which implies that the advertiser can increase their value by increasing their bid while increasing the slack of the ROS constraint.
    This implies that no equilibrium can have $b_1 < \frac{v_1}{1+\ln(\alpha)+\ln(\beta)}$.
    
    We now perform some calculations.
    We have $\frac{\dd}{\dd b} \pi(b) = \frac{1}{2b\ln(\alpha)}$.
    Therefore,
    \[
        \frac{\dd}{\dd b} \Val(b) = \frac{v_1}{2b \ln \alpha}
        \quad \text{and} \quad
        \frac{\dd}{\dd b} \Cost(b) = \frac{1}{2 \ln(\alpha)} + \frac{1}{2} \left( 1 + \frac{\ln(b/b^o)}{\ln(\alpha)} \right).
    \]
    Now, assuming that $v_1 > b_1(1+\ln(\alpha)+\ln(\beta))$, we have
    \begin{align*}
        \frac{\dd}{\dd b} \Val(b_1)
        ~=~ \frac{v_1}{2b_1 \ln(\alpha)}
        ~>~ \frac{1+\ln(\alpha)+\ln(\beta)}{2\ln(\alpha)}
        ~=~ \frac{\dd}{\dd b} \Cost(b_1),
    \end{align*}
    which completes the proof.
\end{proof}

\subsection{Computation of truthful prices for $\randauction(\alpha)$.}
\AppendixName{Rand_alpha_calculation}
In this section, we show that the prices defined in \Definition{Rand_alpha} are exactly the prices obtained using Myerson's payment formula.
Let $b$ denote an advertiser's bid and let $b^o$ denote the other bid.
Let $\beta = b/b^o$. We assume that $\beta \in [1/\alpha, \alpha]$.
Note that if $\beta < 0$ then the advertiser pays $0$ and the advertiser's price remains constant for $\beta \in [\alpha, \infty)$.
Finally, let $\pi(\beta) = \frac{1}{2}\left(1 + \ln_{\alpha} \beta\right)$; this is the probability that the advertiser bidding $b$ wins the query.
Then the price obtained using Myerson's payment formula for bidding $b$ when the other bid is $b^o$
\begin{align*}
    \beta b^o & \cdot \pi(\beta) - \int_{b^o / \alpha}^b \pi(x / b^o) \, \dd x \\
    & ~=~ b^o \left[ \beta \pi(\beta) - \int_{1/\alpha}^{\beta} \pi(u) \, \dd u \right] \\
    & ~=~ b^o \left[ \frac{\beta}{2}\left( 1 + \frac{\ln \beta}{\ln \alpha} \right) - \frac{1}{2}(\beta - \frac{1}{\alpha}) - \frac{1}{2\ln \alpha} \left( \beta \ln \beta - \beta + \frac{1}{\alpha} \ln \alpha + \frac{1}{\alpha} \right) \right] \\
    & ~=~ b^o \cdot \frac{\beta-1/\alpha}{2\ln\alpha}.
\end{align*}

\subsection{Proof of \Claim{rand_poa_lb}}
\AppendixName{rand_poa_lb}
\begin{proof}[Proof of \Claim{rand_poa_lb}]
    Consider the following instance with two advertisers and two queries.
    Let $s = \frac{\alpha-1/\alpha}{2\ln\alpha} \geq 1$ and $\eps < 1/s$.
    \begin{itemize}
        \item For query $1$, advertiser $1$ has value $1$ and advertiser $2$ has value $0$.
        \item For query $2$, advertiser $1$ has value $\eps$ and advertiser $2$ has value $1/s$.
    \end{itemize}
    From \cite{AggarwalBM19}, we know that the optimal bidding strategy for both advertisers is a uniform bidding strategy with multiplier at least $1$.
    We claim that the following bid multipliers form an equilibrium.
    \begin{itemize}
        \item Advertiser $1$ uses a bid multiplier that is at least $\alpha/\eps s$.
        \item Advertiser $2$ uses a bid multiplier equal to $1$.
    \end{itemize}
    In this case, advertiser $1$ gets query $1$ for free and query $2$ at a price of $1$.
    The value for advertiser $1$ is $1+\eps$.
    On the other hand, advertiser $2$ gets no queries and pays nothing.
    So both advertisers satisfy their ROI constraint.
    
    Advertiser $1$ has no incentive to deviate since advertiser $1$ gets both queries.
    We now show that advertiser $2$ also has no incentive to deviate.
    Let $b_1 \geq \alpha / s$ be the bid of advertiser $1$.
    To win query $2$ with non-zero probability, advertiser $2$ must bid at strictly more than $b_1 / \alpha$.
    Suppose that advertiser $2$ bids $\beta b_1$ where $\beta > 1/\alpha$.
    We claim that advertiser $2$ does \emph{not} satisfy their ROS constraint.
    We assume that $\beta \in (1/\alpha, \alpha]$ as the case $\beta > \alpha$ is identical to the case $\beta = \alpha$.
    In this case, advertiser $2$ wins has expected value $\frac{1}{2}\left(1 + \frac{\ln \beta}{\ln \alpha} \right) \cdot \frac{1}{s}$
    and pays $b_1 \cdot \frac{\beta - 1/\alpha}{2 \ln \alpha} \geq \frac{\alpha}{s} \cdot \frac{\beta - 1/\alpha}{2 \ln \alpha}$, in expectation.
    Let us check that
    \begin{equation}
        \EquationName{rvcg_impossible}
        \frac{1}{2} \left(1 + \frac{\ln \beta}{\ln \alpha}\right) < \alpha \cdot \frac{\beta - 1/\alpha}{2 \ln \alpha}
    \end{equation}
    whenever $\beta > 1/\alpha$.
    This implies that advertiser $2$ cannot get any positive value while satisfying their ROS constraint.
    Indeed, rearranging, we get that \Equation{rvcg_impossible} is equivalent to $\ln(\alpha)+\ln(\beta) = \ln(\alpha \beta) < \alpha \beta - 1$, which is true by \Claim{rvcg_numeric_inequality} (since $\alpha \beta > 1$).
\end{proof}

\begin{claim}
    \ClaimName{rvcg_numeric_inequality}
    For $x > 1$, $x-1 > \ln x$.
\end{claim}
\begin{proof}
    The inequality is equivalent to $e^{x-1} > x$ or, replacing $x-1$ with $y$, $e^y > 1+y$ for $y > 0$.
    It is straightforward to see that this last inequality is true.
\end{proof}

\subsection{PoA of $\randauction(\alpha)$}
\AppendixName{randauction_poa}
In this section, we prove a PoA bound for $\randauction(\alpha)$.
In our results, we assume that all advertisers bid at least their value which follows from the charaterization given by \citet{AggarwalBM19}.
\begin{theorem}
\TheoremName{randauction_poa}
Assuming that all advertisers bid at least their value,
$\randauction(\alpha)$ achieves a PoA of $1.897$ for $2$ bidders.
\end{theorem}
Note that \citet{Mehta22} proved a PoA of $1.896$ so their PoA is only marginally better.
\begin{lemma}
    \LemmaName{rand_technical}
    Assuming that all advertisers bid at least their value,
    then for any set of undominated bids for two bidders,
    $\randauction(\alpha)$ obtains at least $f(\alpha)$ fraction of the optimal welfare where
    \begin{equation}
        \EquationName{f_alpha_rand}
        f(\alpha) = \max_{\substack{\gamma, \eta \geq 0 \\ \gamma + \eta = 1}} \min \left\{
            \eta \alpha,
            \gamma,
            \min_{\beta \in [1/\alpha, \alpha]} g(\alpha, \beta, \gamma, \eta)
        \right\} 
    \end{equation}
    and
    \begin{equation}
        \EquationName{g_alpha_beta_gamma_eta_rand}
        g(\alpha, \beta, \gamma, \eta) =
        \frac{\gamma}{2} \left( 1 + \ln(\beta) / \ln(\alpha) \right) +
        \frac{\eta(1-1/\alpha)(1+1/\beta)}{2 \ln \alpha}.
    \end{equation}
\end{lemma}
\begin{proof}
The proof follows a very similar structure as the proof of \Lemma{rfpa_technical}.
Recall that the optimal bidding strategy for the advertisers is to bid some constant multiple $c \geq 1$ of their value.

\paragraph{Case 1: $\bs_j / \bo_j < 1/\alpha$.}
Let $Q_1 = \setst{j \in Q}{\bs_j / \bo_j \leq 1/\alpha}$.
Define $m_1$ as the probability that the optimal bidder is matched to query $j$ whenever $j \in Q_1$.
Furthermore, let $\Spend_1(j)$ be the spend in this query.
Then $m_1 = 0$ and $\Spend_1(j) = \bs_j \cdot \frac{\alpha - 1/\alpha}{2 \ln \alpha} \geq \frac{\alpha - 1/\alpha}{2\ln \alpha} \vs_j$ where the inequality is from \Lemma{wp1_win_lb}.

\paragraph{Case 2: $1/\alpha \leq \bs_j / \bo_j \leq \alpha$.}
For $\beta \in [1/\alpha, \alpha]$, let $Q_2^\beta = \setst{j \in Q}{\bs_j / \bo_j = 1/\beta}$.
Define $m_2^\beta$ as the probability of matching $i^*$ to query $j$ whenever $j \in Q_2^\beta$ and let $\Spend_2^{\beta}(j)$ be the expected spend on query $j$.
Then $m_2^\beta = \frac{1}{2}\left( 1 + \frac{\ln \beta}{\ln \alpha} \right)$ and
\begin{align*}
    \Spend_2^\beta(j)
    & = \bo_j \cdot \frac{\beta - 1/\alpha}{2\ln \alpha}
      + \bs_j \cdot \frac{1/\beta - 1/\alpha}{2 \ln \alpha} \\
    & = \frac{\bs_j}{\beta} \cdot \frac{\beta - 1/\alpha}{2\ln \alpha}
      + \bs_j \cdot \frac{1/\beta - 1/\alpha}{2 \ln \alpha} \\
    & \geq \vs_j \cdot \frac{(1-1/\alpha)(1+1/\beta)}{2 \ln \alpha}.
\end{align*}

\paragraph{Case 3: $\alpha < \bs_j / \bo_j$.}
Finally, let $Q_3 = \setst{j \in Q}{\bs_j / \bo_j \geq \alpha}$.
For any $j \in Q_1$, define $m_3$ as the probability that the optimal bidder is matched in $j$, and $\Spend_3(j)$ be the expected spend in $j$. Then,
\[
    m_3 = 1 \quad \text{and} \quad \Spend_3(j) \geq 0.
\]
The remainder of the proof is identical to that of \Lemma{rfpa_technical}.
\end{proof}
\begin{proof}[Proof of \Theorem{randauction_poa}]
    One can numerically check that $\alpha = 5/9$ gives $f(\alpha) \geq 1/1.897$ by taking $\gamma = 0.528$ in the $\max$.
\end{proof}
\section{Proof of \Theorem{RandPoA}}
\AppendixName{RandPoA}
First, we begin with a lemma that states that if an advertiser is facing $k$ other bids that weakly dominates its own
then it wins with probability $1/(k+1)$.
However, for our proof, we only require such a statement for a specific form of bids.
\begin{lemma}
    Let $b_i \in \bR$ and $b \geq b_i$.
    For any $S \subseteq [n] \setminus \{i\}$,
    $\alloc_i(b \cdot \bme_S + b_i \cdot \bme_i) \leq 1/(|S| + 1)$.
\end{lemma}
\begin{proof}
    Since $\pi_i(b_1, \ldots, b_n)$ is non-decreasing in $b_i$, we have $\alloc_i(b\cdot \bme_S + b_i \cdot \bme_i) \leq \alloc_i(b\cdot \bme_{S \cup \{i\}}) = \alloc_j(b\cdot \bme_{S \cup\{i\}})$ for all $j \in S$ where the last equality is by anonymity.
    As $\sum_{j \in [n]} \alloc_i(b\cdot \bme_{S \cup \{i\}}) \leq 1$, this implies that $\alloc_i(b\cdot \bme_S + b_i \cdot \bme_{\{i\}}) \leq 1/(|S| + 1)$.
\end{proof}

For the remainder of this section, we assume that there are $2k$ bidders, for some integer $k \geq 1$ and we fix an auction with allocation rule $\alloc$ and pricing function $\price$ that satisfies Assumption~\ref{assumption}.
Further, fix a constant $\gamma > 0$.
To construct out instance, we first define a number of parameters.
Let $\eps > 0$ be a constant.
\begin{enumerate}
    \item Let $\alloc^* = \lim_{b \to \infty} \pi_1(b \cdot \bme_1)$.
    We assume that $\alloc^* \geq 1/2$ (otherwise the auction has a trivial PoA of at least $2$ from just a single bidder).
    \item Let $B_1 \in \bR$ be such that $\pi_1(B_1 \bme_1) > \alloc^* - \eps$.
    Thus, $B_1$ is the bid that a bidder needs to place to win the item with maximum probability (up to an additive $\eps$) if no other bidders place a positive bid.
    \item Define the function $M(b) = \price_1(b\cdot \bme_{[k+1]}) / \alloc_1(B_1 \cdot \bme_1) + \price_1(B_1 \cdot \bme_1)$.
    Note that $\lim_{b \to \infty} M(b) = +\infty$ by the pricing assumption in Assumption~\ref{assumption}.
\end{enumerate}
We are now ready to construct our lower bound instance.
Our construction has $2k$ queries which we denote by $q_1, \ldots, q_{2k}$.
These queries come in two types.
\begin{enumerate}
    \item For query $q_i$ where $i \in [k]$, bidder $i$ has value $M(B_2)$ and all other bidders have value $0$.
    \item For query $q_{k+i}$ where $i \in [k]$, bidder $k+i$ has value $(M(B_2) - \price_1(B_1 \bme_1)) \cdot \pi_1(B_1 \bme_1) - \eps$.
    On the other hand, bidders $1, \ldots, k$ have value $\eps$.
\end{enumerate}

\begin{proof}[Proof of \Theorem{RandPoA}]
    Fix $0 < \eps < \min\{1/4, \gamma / 8\}$.
    Choose $B_1$ (depending on $\eps$) as in the aforementioned construction and let $B_2$ be such that $M(B_2) \geq k$.
    We relegate the choices of $\eps, B_2$ to the end of the proof to achieve the desired result.

    Consider the following bids.
    \begin{enumerate}
        \item For $i \in [k]$, bidder $i$ bids $B_1$ on query $q_i$ and all other bidders bid $0$.
        \item For $i \in [k]$, bidders $1, \ldots, k$ bid $B_2$, bidder $k+i$ bids at most $B_2$, and all other bidders bid $0$.
    \end{enumerate}
    
    \paragraph{Equilibrium analysis.}
    We first check that for $i \in [k]$, bidder $k+i$ has no incentive to bid more than $B_2$.
    Indeed, if bidder $k+i$ does bid more than $B_2$ then its payment (when it wins) is at least $\price_1(B_2 \cdot \bme_{[k+1]})$
    (which follows from the assumption that prices are non-decreasing in the bids).
    By definition of $M(B_2)$, we have that $\price_1(B_2 \cdot \bme_{[k+1]}) = (M - \price_1(B_1 \cdot \bme_1)) \cdot \pi_1(B_1 \cdot \bme_1)$
    and this is strictly more than the value that bidder $k+i$ has for query $q_{k+i}$.
    
    Now we check that for $i \in [k]$, bidder $i$ has no incentive to change its bid to gain more than $\gamma$ in total value.
    For $i \in [k]$, bidder $i$ wins query $i$ with probability at least $\alloc^* - \eps$.
    So its utility is at least $M(B_2) \pi_1(B_1 \cdot \bme_1) \geq M(B_2) \cdot (\alloc^* - \eps)$.
    The maximum utility it can achieve is upper bounded by $M(B_2) \cdot \alloc^* + k\eps$.
    Note that (since we assume $M(B_2) \geq k$ and $\pi^* \geq 1/2$)
    \[
        \frac{M(B_2) \cdot \pi^* + k \eps}{M(B_2) \cdot (\pi^* - \eps)} 
        \leq \frac{\pi^*+\eps}{\pi^* - \eps}
        \leq \frac{1/2 + \eps}{1/2 - \eps}
        = \frac{1+2\eps}{1-2\eps}
        \leq (1+2\eps)(1+4\eps)
        = 1 + 6 \eps + 8 \eps^2
        \leq 1 + 8 \eps
        \leq 1 + \gamma,
    \]
    where the third inequality is \Claim{RecipIneq}, the fourth inequality is the assumption that $\eps \leq 1/4$
    and the last inequality is the assumption that $\eps \leq \gamma / 8$.
    In particular, bidder $i$ cannot increase its value by a factor of $(1+\gamma)$.
    
    Next, we check that their ROS constraint is also not violated.
    Indeed, the total payment made by advertiser $i$ is at most
    \begin{align*}
        \price_1(B_1 \cdot \bme_1) & + \price_1(b \cdot \bme_{[k+1]}) \\
        & = \price_1(B_1 \cdot \bme_1) \pi_1(B_1 \cdot \bme_1)
        + (M(B_2) - \price_1(B_1 \cdot \bme_1)) \cdot \pi_1(B_1 \cdot \bme_1) \\
        & = M(B_2) \pi_1(B_1 \cdot \bme_1),
    \end{align*}
    which is the value that bidder $i$ receives from the first query alone.
    So their ROS constraint is satisfied.
    
    \paragraph{Liquid welfare and PoA.}
    Let $i \in [k]$.
    In the equilibrium above, the contribution from advertiser $i$ to the liquid welfare is bounded above by $M(B_2) \cdot \pi^* + k\eps$
    while the contribution from advertiser $k+i$ to the liquid welfare is bounded above by $\frac{M(B_2) \cdot \pi^*}{k+1}$.
    Let $\LW(EQ)$ be the liquid welfare of this equilibrium.
    Then, we have
    \begin{equation}
        \EquationName{randomized_lb1}
        \LW(EQ) \leq k \cdot M(B_2) \cdot \pi^* + k^2\eps + \frac{k M(B_2) \cdot \pi^*}{k+1}.
    \end{equation}
    On the other hand, the optimal allocation is to allocate query $j$ to bidder $j$ for $j \in [2k]$ for a welfare of
    \begin{equation}
        \EquationName{randomized_lb2}
        \LW(OPT) = k(1 + \pi_1(B_1 \bme_1)) \cdot M(B_2) - k\pi_1(B_1 \cdot \bme_1) \price_1(B_1 \cdot \bme_1) - k\eps. 
    \end{equation}
    We now set parameters.
    Let $\delta_1 = \min\{1/2, \delta / 11\}$.
    We set $\eps = \min\{ 1/4, \gamma / 8, 2\delta_1 / (1+\delta_1) \}$ so that $1/(2-\eps) \leq (1+\delta_1) / 2$.
    Having chosen $\eps$, we can fix $B_1$ such that $\pi_1(B_1 \bme_1) > \pi^* - \eps$.
    We now choose $B_2$ sufficiently large such that
    \[
        \LW(EQ) \leq k M(B_2) \pi^* \cdot \left(1 + \frac{1}{k+1} \right) (1 + \delta_1)
    \] 
    and
    \[
        \LW(OPT) \geq k(1 + \pi_1(B_1 \cdot \bme_1)) M(B_2) \cdot (1 - \delta_1).
    \]
    Thus, we have
    \begin{align*}
        \frac{\LW(EQ)}{\LW(OPT)}
        & \leq \frac{\pi^*(1+1/(k+1))}{1+\pi_1(B_1 \cdot \bme_1)} \cdot \frac{1+\delta_1}{1-\delta_1} \\
        & \leq \frac{\pi^*(1+1/(k+1))}{1+\pi^* - \eps} \cdot \frac{1+\delta_1}{1-\delta_1} \\
        & \leq \frac{1+1/(k+1)}{2 - \eps} \cdot \frac{1+\delta_1}{1-\delta_1} \\
        & \leq \frac{1+1/(k+1)}{2} \cdot \frac{(1+\delta_1)^2}{1-\delta_1} \\
        & \leq \frac{1+1/(k+1)}{2} \cdot (1+\delta_1)^2(1+2\delta_1) \\
        & \leq \frac{1+1/(k+1)}{2} \cdot (1+\delta).
    \end{align*}
    The second inequality is by definition of $B_1$,
    the third inequality is because the second line is maximized at $\pi^* = 1$,
    the fourth inequality is by our choice of $\eps$,
    the fifth inequality is by \Claim{RecipIneq},
    and the last inequality uses that $\delta_1 \leq 1/2$ and $\delta_1 \leq \delta / 11$.
    Multiplying top and bottom by $k+1$ gives the expression in the claim.
\end{proof}

\begin{claim}
    \ClaimName{RecipIneq}
    Suppose $x \in [0, 1/2]$.
    Then $\frac{1}{1-x} \leq 1+2x$.
\end{claim}
\begin{proof}
    If $x \in [0, 1/2]$ then $(1+2x)(1-x) = 1 + x - 2x^2 \geq 1 + x - x = 1$.
\end{proof}
\end{document}